%% file: main.tex
\documentclass{article}

\usepackage{microtype}
\usepackage{graphicx}
\usepackage{booktabs}
\usepackage{hyperref}
\usepackage[accepted]{icml2026}
\usepackage{xcolor}
\usepackage{tikz}
\usetikzlibrary{arrows.meta,fit,backgrounds,positioning}
\usepackage[above]{placeins}
\usepackage{amsmath}
\usepackage{amssymb}
\usepackage{mathtools}
\usepackage{amsthm}
\usepackage{enumitem}
\usepackage[capitalize,noabbrev]{cleveref}
\usepackage{xspace}

\input{math_commands}

\theoremstyle{plain}
\newtheorem{theorem}{Theorem}[section]
\newtheorem{proposition}[theorem]{Proposition}

\icmltitlerunning{ChainCaps}

\begin{document}

\twocolumn[
\icmltitle{ChainCaps: Composition-Safe Tool-Using Agents via Monotonic Capability Attenuation}

\begin{icmlauthorlist}
\icmlauthor{Xiaochong Jiang}{indep1}
\icmlauthor{Shiqi Yang}{indep2}
\icmlauthor{Ziwei Li}{King Abdullah University of Science and Technology}
\icmlauthor{Lifei Liu}{indep1}
\icmlauthor{Haoran Yu}{indep1}
\icmlauthor{Yichen Liu}{indep1}
\end{icmlauthorlist}

\icmlaffiliation{indep1}{Independent Researcher, Seattle, WA, USA}
\icmlaffiliation{indep2}{Independent Researcher, New York City, NY, USA}
\icmlaffiliation{King Abdullah University of Science and Technology}{King Abdullah University of Science and Technology, Saudi Arabia}

\icmlcorrespondingauthor{Xiaochong Jiang}{jiang.xiaoc@northeastern.edu}
\icmlcorrespondingauthor{Shiqi Yang}{sy3506@nyu.edu}
\icmlcorrespondingauthor{Ziwei Li}{ziwei.li@kaust.edu.sa}
\icmlcorrespondingauthor{Lifei Liu}{lliu.lifei@gmail.com}
\icmlcorrespondingauthor{Haoran Yu}{haoranyu889@gmail.com}
\icmlcorrespondingauthor{Yichen Liu}{yil160@ucsd.edu}

\icmlkeywords{LLM agents, agent safety, tool security, information flow control, composition safety, MCP}

\vskip 0.3in
]

\printAffiliationsAndNotice{}

\input{sections/0_abstract}
\input{sections/1_introduction}
\input{sections/2_related_work}
\input{sections/3_method}

\input{sections/4_evaluation}
\input{sections/5_conclusion}

\bibliography{references}
\bibliographystyle{icml2026}

\end{document}

%% file: math_commands.tex
\newcommand{\chaincaps}{\mbox{ChainCaps}\xspace}

%% file: sections/0_abstract.tex
\begin{abstract}
Tool-using agents increasingly operate in open-ended deployment environments, where they compose file systems, web APIs, code interpreters, and enterprise services at runtime. This creates a safety gap in tool composition: an agent can satisfy every per-tool permission check and still produce an unsafe end-to-end effect, such as reading a confidential document, summarizing it, and sending the summary to an external endpoint. We call this failure mode \emph{permission laundering}. \chaincaps addresses it with a runtime rule: every value carries a sink-specific capability budget, and tool composition propagates budgets by intersection. A value can preserve or lose authority as it moves through a tool chain, but it cannot gain new authority through composition. We implement \chaincaps as a transparent MCP proxy that requires no changes to the agent or tool servers. On 82 tasks across five frontier models from three providers, \chaincaps reduces attack success rate from 25--68\% to 0--4.8\% while preserving 96--100\% benign completion. In replay experiments, it also outperforms scalar-IFC and per-function-isolation baselines. Manifest quality is the dominant deployment bottleneck: expert manifests reach 100\% attack blocking, while naive manifests fall to 27.3\%. Our claims are limited to explicit-flow composition safety under trusted manifests and proxy-visible data movement, a practical gap in deployed tool-using agents today.
\end{abstract}

%% file: sections/1_introduction.tex
\section{Introduction}
\label{sec:introduction}

Consider a deployed agent asked to prepare an executive summary. It reads a salary spreadsheet (\texttt{read\_file}), fetches industry benchmarks from the web (\texttt{fetch\_url}), summarizes the combined data (\texttt{summarize}), and emails the result to a distribution list (\texttt{send\_email}). Each call is individually authorized. The composition is not: salary data that should only be used internally has now been included in an outbound message. No single tool misbehaved; the end-to-end workflow violated policy.

Tool-using language-model agents increasingly run real workflows in exactly this pattern. In deployed settings, they read local files, query internal services, call web APIs, and write results back through the Model Context Protocol (MCP) or similar interfaces \cite{mcp2024,seagent2026,pfi2025}. Adjacent work on human-preference prediction further shows that deployed LLM pipelines can face performance--efficiency trade-offs when choosing between fine-tuned LLMs and lighter traditional models \cite{zhang2026performance}. Roadmaps for AI-assisted research likewise frame future systems around multi-step workflows that coordinate retrieval, analysis, and tool use \cite{kong2026aiautoresearch}. High-stakes medical vision-language systems show a similar move toward iterative diagnostic state, including dynamic visual focus for progressive diagnosis and latent memory for clinical reasoning \cite{zhu2026medeyes,zhu2026medsynapsevbridgingvisualperception}. The runtime call graph is assembled by the model at inference time. Similar routing pressure appears in retrieval-heavy applications: recent financial RAG systems route queries to whole documents before scoped chunk retrieval to reduce cross-document confusion \cite{cheng2026hdrr}. This makes agents useful for flexible workflows, but it also exposes a safety gap that per-tool authorization does not address: the composition of locally acceptable calls can violate a deployment policy.

We call this failure mode \emph{permission laundering}. The agent does not need to violate any single tool's local permission rule. Instead, authority is effectively recovered through a sequence of allowed calls: a restricted value is transformed, mixed with other values, or rewritten before reaching a sink it should not reach. For deployed agents, this is a safety problem as much as a security problem, because it can arise in ordinary workflows rather than only under overtly malicious use.

This gap is not hypothetical. ChainFuzzer reports 365 multi-tool vulnerabilities across 19 of 20 real-world agent applications \cite{chainfuzzer2026}. Models dynamically route intermediate values across tools in ways that static planners, server-side access checks, and per-function wrappers do not capture. The exact composition depends on prompt wording, retrieved context, partial outputs, and tool return values, so the unsafe flow is often not known until the workflow actually runs. A deployment mechanism must reason about authority propagation over the actual dataflow that occurs, not just about the access rights of individual tools in isolation.

Existing defenses only partially cover this setting. Label-based information-flow defenses such as Fides track confidentiality and integrity at the value level, but scalar labels do not directly encode which downstream sinks a value may still reach \cite{fides2025}. Isolation-based approaches such as PFI and server-level permission enforcement reduce blast radius for individual calls or servers, but they do not by themselves prevent safe-looking calls from composing into unsafe flows \cite{pfi2025,agentbound2025}. Prompt-injection mitigations and safety monitors can help detect attacks or trigger refusals, but they usually do not enforce a runtime invariant over authority propagation through a chain of tool calls \cite{shieldagent2025,mite2026}.

We present \chaincaps, a runtime mechanism for composition-safe tool use. The key idea is to attach a sink-specific authority budget to every value and propagate that budget monotonically through the workflow. If an intermediate value depends on multiple inputs, its future authority becomes the intersection of the inputs' remaining authority and the tool's declared pass-through budget. Composition can therefore preserve or reduce a value's downstream authority, but it cannot widen what the value may do next. We implement this as an MCP proxy, so it can sit between frontier-model agents and existing tool servers without changing either side.

Our evaluation targets the deployment setting directly. Across an 82-task stress-test suite, five frontier models, and three runs per condition, \chaincaps reduces attack success rate from 25--68\% without defense to 0--4.8\% while preserving 96--100\% benign completion. In replay-based comparison, it blocks substantially more attack traces than modeled scalar-IFC and per-function-isolation baselines under the same traces and metadata. The main deployment lesson is that manifest quality matters: with expert manifests, blocking reaches 100\% on the manifest-quality sweep, while naive manifests drop to 27.3\%.

This paper makes four contributions:
\begin{enumerate}[leftmargin=1.25em,itemsep=0.2em,topsep=0.3em]
\item We identify \emph{permission laundering} as a runtime composition failure in tool-using agents: locally authorized calls can compose into an unauthorized end-to-end effect.
\item We formulate composition safety as monotonic budget propagation and state a non-amplification theorem: without declassification, composition cannot grant sink authority absent from the contributing sources' initial budgets.
\item We describe a transparent MCP proxy implementation of \chaincaps, including manifests, signed one-shot declassification tokens, fail-closed handling for unknown sinks, session isolation, and a manifest linter.
\item We evaluate \chaincaps on live frontier-model agents and replay traces, and identify manifest quality as the main deployment bottleneck rather than runtime overhead or model compatibility.
\end{enumerate}

Our scope is explicit-flow composition safety under trusted manifests and intact proxy enforcement. \chaincaps does not model covert channels, implicit flows through hidden model state, latent diagnostic memory inside multimodal models \cite{zhu2026medsynapsevbridgingvisualperception}, or policy errors in the manifests themselves. These limits are important, but they do not remove the need for runtime enforcement over dynamically assembled tool chains. \chaincaps provides one practical way to ensure that data can lose authority as it moves through a workflow, but cannot gain authority through composition.

%% file: sections/2_related_work.tex
\section{Related Work}
\label{sec:related}

\paragraph{Information-flow control for LLM systems.}
Fides attaches confidentiality and integrity labels to agent state and tool I/O \cite{fides2025}. The LLMbda Calculus formalizes conversations, state, and flow tracking for agent-like systems \cite{llmbda2026}. Both establish that value-level flow tracking is the right abstraction for language-model agents. Our setting differs in what the policy must express: in deployed tool chains, the relevant question is often not whether a value is ``secret'' but \emph{which downstream sinks it may still reach}. A value may be displayable in one UI, writable in one directory, or sendable only to one URL prefix. \chaincaps tracks sink-specific budgets directly. In our replay study, this sink-specific representation blocks substantially more traces than a scalar-label baseline instantiated on the same logs.

\paragraph{Execution isolation and environment controls.}
SEAgent combines mandatory access control with flow-graph reasoning to reduce privilege escalation \cite{seagent2026}. PFI isolates tool invocations so prompt flow cannot directly widen privileges across functions \cite{pfi2025}. AgentBound pushes permissions down to server or execution boundaries \cite{agentbound2025}. These mechanisms address failures that an application-layer proxy cannot see, including shell-level or OS-level side effects, and are valuable for defense in depth. Their focus is different from ours: they constrain execution boundaries, while \chaincaps constrains the authority of \emph{derived values} as they move through a workflow. The two approaches are complementary.

\paragraph{Detection, testing, and monitoring.}
ShieldAgent reasons over safety policies during execution \cite{shieldagent2025}. MCP-Guard provides defense in depth for the MCP ecosystem \cite{mcpguard2025}. ChainFuzzer and AgentGuard expose unsafe tool compositions through fuzzing and orchestrator-based safety evaluation \cite{chainfuzzer2026,agentguard2025}. Complementary benchmark-reliability work studies configuration-conditional rank instability on alignment benchmarks \cite{li2026safetyreproconfigurationconditionalrankinstability}. These systems help operators discover vulnerabilities and benchmark agents. We use ChainFuzzer to motivate the prevalence of multi-tool failures, and replay attacks adapted from InjecAgent and ToolEmu to test whether the enforcement rule transfers beyond our in-house tasks \cite{injecagent2024,toolemu2024}. But discovery is not enforcement: a monitor may detect a bad chain after it forms, while \chaincaps prevents the chain from executing its harmful sink in the first place when the sink remains visible to the proxy.

\paragraph{Broader agent safety.}
Prompt injection, indirect injection, adversarial prompt safeguards, retrieval poisoning, and associated operational controls are studied in \cite{mite2026,lin2026reflect}. These defenses are complementary, not competing: our residual failures in script indirection and shell-pipe laundering arise precisely because those cases fall outside an application-layer proxy's visibility. The invariant \chaincaps enforces---composition cannot create sink authority absent from all contributing inputs---is what existing safety layers typically leave implicit.

%% file: sections/3_method.tex
\section{Monotonic Capability Attenuation}
\label{sec:method}

\subsection{Threat Model and Deployment Model}
\label{sec:threat-model}

We consider a tool-using agent connected to multiple MCP servers or similar tool interfaces. The adversary may influence user prompts, retrieved documents, web content, or intermediate tool outputs. The adversary's goal is to cause a value derived from a restricted source to reach an unauthorized sink through a sequence of locally plausible tool calls. \chaincaps assumes that relevant tool calls pass through the proxy, manifests correctly identify sources and sinks, and the data movement to be controlled is visible at the protocol boundary. We do not claim protection against covert channels, hidden model-state leakage, compromised tool servers, policy errors in manifests, or OS-level side effects that bypass the proxy. In particular, a malicious tool server that hides an effect or declares an overly permissive $\mathrm{Pass}(t)$ can invalidate the deployment policy; our linter helps catch accidental manifest errors, but adversarial tool semantics require review, sandboxing, or lower-level isolation. These assumptions scope the guarantee to explicit-flow composition safety, which is the deployment gap targeted in this paper.

\paragraph{Intuition.}
\chaincaps enforces one simple rule: every value carries a budget, meaning the set of sinks it may still reach. When tools transform or combine values, the output budget is the intersection of the inputs' budgets and the tool's pass-through policy. The effect is monotonic capability attenuation. A value can keep authority or lose authority, but it cannot gain new authority just because the agent routed it through more tools.

\input{figures/budget-propagation}

\subsection{Budget Algebra}
\label{sec:algebra}

We model a sink privilege as a pair
\[
s = (\mathit{op}, \mathit{scope}),
\]
where $\mathit{op}$ names an effectful operation such as \texttt{http\_send}, \texttt{file\_write}, or \texttt{exec}, and $\mathit{scope}$ restricts where the effect may occur, such as a URL prefix, filesystem path, or command family. Privileges are ordered by scope inclusion:
\begin{equation}
\label{eq:order}
(\mathit{op}_1,\sigma_1) \preceq (\mathit{op}_2,\sigma_2)
\iff
\mathit{op}_1 = \mathit{op}_2
\wedge
\sigma_1 \subseteq \sigma_2.
\end{equation}
A narrower scope is therefore a weaker privilege.

A budget $B$ is a downward-closed set of sink privileges. Downward-closed means that if $s \in B$ and $s' \preceq s$, then $s' \in B$ as well. Operationally, $B$ is the set of sinks the value may still reach. Each source origin $o$ is assigned an initial budget $\mathrm{Init}(o)$, which specifies the maximum authority any value from that source may ever have.

Each tool manifest declares three pieces of metadata. First, $\mathrm{Exec}(t)$ is the family of sink privileges the tool may exercise. Second, $\mathrm{Pass}(t)$ is an upper bound on the budget of values produced by the tool. Third, the runtime derives a call-specific sink requirement $\mathrm{Req}(t,a)$ from the invoked arguments $a$, such as ``send to \texttt{https://api.example.com/v1/}'' or ``write under \texttt{/tmp/reports/}.'' In the MCP setting, these fields are tool-local and protocol-independent: they describe the semantics of the tool, not the wire format used to call it. This design keeps policy at the deployment boundary where operators already reason about source roles, sink scopes, and tool effects.

\subsection{Runtime Enforcement}
\label{sec:runtime}

\paragraph{Source initialization.}
When a source value $v$ enters from origin $o$, the proxy sets
\[
B(v) := \mathrm{Init}(o).
\]
Examples include a file read, a secret fetched from a credential store, or a value returned by a database query.

\paragraph{Transfer rule.}
When tool $t$ consumes values $x_1,\dots,x_k$ and produces output $y$, the runtime sets
\begin{equation}
\label{eq:transfer}
B(y) = \mathrm{Pass}(t) \cap B(x_1) \cap \cdots \cap B(x_k).
\end{equation}
This meet rule is the core monotonicity property. If a restrictive value is mixed with a permissive value, the result inherits the restrictive budget. That is exactly what we want for deployed agents: summarizing, formatting, or concatenating data should not recover authority that one of the ingredients already lacked.

\paragraph{Sink rule.}
Before an effectful invocation $(t,a)$, the proxy computes the set of values $D(a)$ that contributed to the sink arguments using the online dataflow DAG and conservative context tracking. The call is allowed only if the required sink remains authorized:
\begin{equation}
\label{eq:sink}
\mathrm{Req}(t,a) \in \bigcap_{x \in D(a)} B(x).
\end{equation}
Otherwise the proxy blocks the call unless a valid declassification token is present.

\paragraph{Context tracking.}
Generated tool arguments often depend on prompt context rather than on one explicit value. To stay sound at the protocol boundary, \chaincaps maintains a context budget $B_{\mathrm{ctx}}$ and intersects it with every value inserted into the model context. Values later recovered from that context inherit $B_{\mathrm{ctx}}$. This is conservative: the proxy does not know which tokens actually caused a future model output. The ablation shows why the conservatism is necessary in this setting; removing context tracking reduces attack blocking by 7 percentage points.

\subsection{Non-Amplification Guarantee}
\label{sec:guarantee}

\begin{theorem}[Non-amplification]
\label{thm:nonamp}
Assume a sink invocation $(t,a)$ is allowed by \textnormal{\chaincaps} without declassification. Let $O(a)$ be the set of source origins that transitively contribute to the sink arguments. Then
\[
\mathrm{Req}(t,a) \in \mathrm{Init}(o)
\qquad
\text{for every } o \in O(a).
\]
\end{theorem}

\begin{proof}[Proof sketch]
Source values satisfy the claim by construction. For any derived value, \cref{eq:transfer} updates its budget only by intersection, so no step can add a sink privilege absent from one of its inputs. Context tracking narrows budgets further. Therefore every value that contributes to $a$ has a budget contained in the intersection of the initial budgets of its contributing sources. If the sink call is allowed, \cref{eq:sink} implies that $\mathrm{Req}(t,a)$ lies in that intersection, and hence in every contributing source budget.
\end{proof}

This theorem is the formal version of composition safety in our setting. Under the stated visibility and manifest assumptions, a chain of individually allowed tool calls cannot create downstream authority that the contributing sources never had.

\subsection{Why Label-Only IFC Is Not Enough}
\label{sec:ifc}

\begin{proposition}
\label{prop:labels}
Let there be $m$ independent sink classes. Any label-only system that exactly distinguishes all possible authorized sink subsets requires at least $2^m$ distinct labels, while a sink-budget representation requires only an $m$-bit membership vector.
\end{proposition}

\begin{proof}[Proof sketch]
Distinct sink subsets must remain distinguishable if enforcement is to treat them differently. There are $2^m$ such subsets, so any exact scalar-label encoding needs at least that many distinct states. A budget stores one bit per sink class, which yields the same policy space directly.
\end{proof}

This is a representation argument, not a dismissal of IFC. Rich IFC systems can encode many policies. The narrower point is that one scalar label per value is an awkward abstraction once deployment policy depends on subsets of downstream sinks rather than on one total order of sensitivity.

\subsection{Implementation}
\label{sec:implementation}

Our reference implementation is an MCP proxy of roughly 1{,}200 lines of Python. The proxy sits between the agent and ordinary MCP servers, intercepts tool calls and results, and maintains three runtime data structures: a budget map from value identifiers to budgets, an online dataflow DAG for lineage recovery, and the context budget. Because enforcement happens at the protocol boundary, neither the frontier model nor the tool servers need modification.

Each manifest declares source budgets, executable sink privileges, and pass-through budgets. Declassification uses HMAC-SHA256 signed one-shot tokens bound to the specific sink request and lineage, so replayed or forged tokens fail verification. For partially specified ecosystems, the proxy fails closed on unknown sinks using a keyword heuristic over tool names and descriptions. In our corpus study this heuristic achieved 100\% sink recall, but with many false positives, so it is a bootstrapping aid rather than a substitute for manifests. We also support session isolation, which resets the context budget at task boundaries and reduces false positives on hard benign scenarios from 45\% to 20\%. Finally, a six-rule manifest linter flags common authoring errors such as wildcard scopes, missing source coverage, and mismatches between declared tool role and observable effect. The evaluation shows that this authoring layer is the main practical bottleneck.

\input{figures/architecture}

%% file: figures/budget-propagation.tex
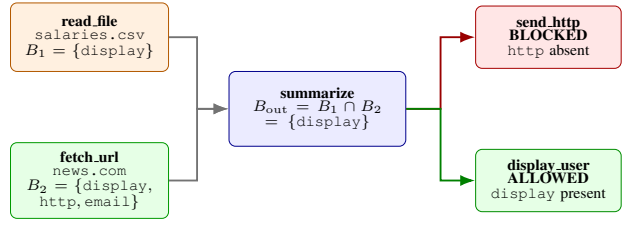
\begin{figure}[t]
\centering
\resizebox{0.98\columnwidth}{!}{%
\begin{tikzpicture}[
  font=\scriptsize,
  >=Latex,
  srcbox/.style={
    draw=orange!70!black, rounded corners=3pt,
    fill=orange!12, text width=2.25cm,
    minimum height=1.10cm, align=center, inner sep=4pt,
  },
  pubbox/.style={
    draw=green!60!black, rounded corners=3pt,
    fill=green!10, text width=2.25cm,
    minimum height=1.10cm, align=center, inner sep=4pt,
  },
  toolbox/.style={
    draw=blue!55!black, rounded corners=3pt,
    fill=blue!8, text width=2.55cm,
    minimum height=1.20cm, align=center, inner sep=4pt,
  },
  outbox/.style={
    draw=orange!70!black, rounded corners=3pt,
    fill=orange!12, text width=2.1cm,
    minimum height=1.00cm, align=center, inner sep=4pt,
  },
  blockbox/.style={
    draw=red!65!black, rounded corners=3pt,
    fill=red!10, text width=2.05cm,
    minimum height=1.00cm, align=center, inner sep=4pt,
  },
  allowbox/.style={
    draw=green!55!black, rounded corners=3pt,
    fill=green!10, text width=2.05cm,
    minimum height=1.00cm, align=center, inner sep=4pt,
  },
  fl/.style={->, line width=0.85pt, draw=black!55},
  denyfl/.style={->, line width=0.85pt, draw=red!60!black},
  allowfl/.style={->, line width=0.85pt, draw=green!50!black},
]

\node[srcbox] (src1) at (0, 1.15) {%
  \textbf{read\_file}\\[-2pt]
  \texttt{salaries.csv}\\[-1pt]
  {$B_1=\{\texttt{display}\}$}};

\node[pubbox] (src2) at (0,-1.15) {%
  \textbf{fetch\_url}\\[-2pt]
  \texttt{news.com}\\[-1pt]
  {$B_2=\{\texttt{display},$}\\[-1pt]
  {$\texttt{http},\texttt{email}\}$}};

\node[toolbox] (sum) at (3.65, 0) {%
  \textbf{summarize}\\[-2pt]
  {$B_{\mathrm{out}}=B_1\cap B_2$}\\[-1pt]
  {$=\{\texttt{display}\}$}};

\node[blockbox] (blk) at (7.35, 1.15) {%
  \textbf{send\_http}\\[-2pt]
  \textbf{BLOCKED}\\[-1pt]
  {\texttt{http} absent}};
\node[allowbox] (alw) at (7.35,-1.15) {%
  \textbf{display\_user}\\[-2pt]
  \textbf{ALLOWED}\\[-1pt]
  {\texttt{display} present}};

\draw[fl] (src1.east) -- ++(0.45,0) |- (sum.west);
\draw[fl] (src2.east) -- ++(0.45,0) |- (sum.west);
\draw[denyfl] (sum.east) -- ++(0.55,0) |- (blk.west);
\draw[allowfl] (sum.east) -- ++(0.55,0) |- (alw.west);

\end{tikzpicture}%
}
\caption{Budget propagation example. A summary combining salary data (display-only) and public news inherits the most restrictive budget via intersection. Because the resulting budget permits display but not HTTP sending, the outbound call is blocked while user display is allowed. This monotonic attenuation is the core runtime property of \chaincaps.}
\label{fig:budget-example}
\end{figure}

%% file: figures/architecture.tex
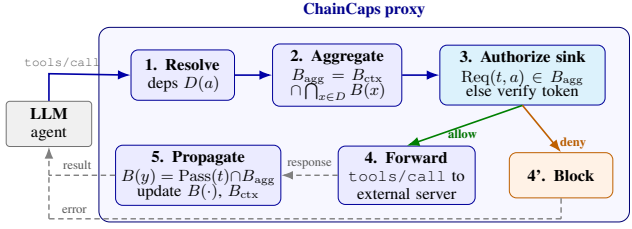
\begin{figure}[t]
\centering
\resizebox{\columnwidth}{!}{%
\begin{tikzpicture}[
  font=\footnotesize,
  >=Latex,
  actor/.style={
    draw=black!55,
    rounded corners=2pt,
    fill=black!6,
    text width=1.45cm,
    minimum height=0.88cm,
    align=center,
    inner sep=3pt,
  },
  stepbox/.style={
    draw=blue!60!black,
    rounded corners=4pt,
    fill=blue!8,
    text width=2.5cm,
    minimum height=0.98cm,
    align=center,
    inner sep=3pt,
  },
  focusbox/.style={
    stepbox,
    fill=cyan!12,
    text width=2.85cm,
  },
  denybox/.style={
    draw=orange!75!black,
    rounded corners=4pt,
    fill=orange!10,
    text width=1.75cm,
    minimum height=0.88cm,
    align=center,
    inner sep=3pt,
  },
  proxybox/.style={
    draw=blue!45!black,
    rounded corners=6pt,
    fill=blue!3,
    inner sep=9pt,
  },
  branchflow/.style={line width=0.9pt, draw=black!45},
  reqflow/.style={->, line width=0.9pt, draw=blue!65!black},
  allowflow/.style={->, line width=0.9pt, draw=green!50!black},
  denyflow/.style={->, line width=0.9pt, draw=orange!75!black},
  respbranch/.style={line width=0.78pt, draw=black!50, densely dashed},
  respflow/.style={->, line width=0.78pt, draw=black!50, densely dashed},
  lab/.style={font=\scriptsize, text=black!65},
  allowlab/.style={font=\scriptsize\bfseries, text=green!45!black},
  denylab/.style={font=\scriptsize\bfseries, text=orange!72!black},
]

\node[actor] (agent) at (-5.6, 0) {\textbf{LLM}\\agent};

\node[stepbox, text width=1.8cm] (s1) at (-3.0, 0.95)
  {\textbf{1. Resolve}\\deps $D(a)$};
\node[stepbox, text width=2.3cm] (s2) at (0.0, 0.95)
  {\textbf{2. Aggregate}\\[1pt]%
   {\footnotesize $B_\mathrm{agg}=B_\mathrm{ctx}$}\\[-2pt]%
   {\footnotesize $\cap\bigcap_{x\in D}B(x)$}};
\node[focusbox, text width=3cm] (s3) at (3.6, 0.95)
  {\textbf{3. Authorize sink}\\[1pt]%
   {\footnotesize $\mathrm{Req}(t,a)\in B_\mathrm{agg}$}\\[-2pt]%
   {\footnotesize else verify token}};

\node[stepbox, text width=3.0cm] (s5) at (-2.7, -1.0)
  {\textbf{5. Propagate}\\[1pt]%
   {\footnotesize $B(y)=\mathrm{Pass}(t)\cap B_\mathrm{agg}$}\\[-2pt]%
   {\footnotesize update $B(\cdot)$, $B_\mathrm{ctx}$}};
\node[stepbox, text width=2.45cm] (s4) at (1.35, -1.0)
  {\textbf{4. Forward}\\{\footnotesize \texttt{tools/call} to external server}};
\node[denybox, text width=1.8cm] (block) at (4.35, -1.0)
  {\textbf{4'. Block}\\};

\path (s3.south) ++(0,0) coordinate (branch);

\begin{scope}[on background layer]
\node[proxybox, fit=(s1)(s2)(s3)(block)(s4)(s5),
  label={[yshift=2pt, font=\small\bfseries, text=blue!55!black]north:ChainCaps proxy}]
  (proxy) {};
\end{scope}

\path (agent.north) ++(0,0.53) coordinate (callrise);
\node[lab, anchor=west] at (-6.2, 1.2) {\texttt{tools/call}};
\draw[reqflow] (agent.north) -- (callrise) -- (s1.west);
\draw[reqflow] (s1.east) -- (s2.west);
\draw[reqflow] (s2.east) -- (s3.west);
\draw[branchflow] (s3.south) -- (branch);
\draw[allowflow] (branch) -- node[allowlab, right, pos=0.7] {allow} (s4.north);
\draw[denyflow] (branch) -- node[denylab, right, pos=0.8] {deny} (block.north);

\draw[respflow] (s4.west) -- node[lab, above] {response} (s5.east);
\path (s5.west) ++(-0.82,0) coordinate (resultbend);
\path (block.south) ++(0,-0.4) coordinate (errorbend);
\path (agent.south |- errorbend) coordinate (retbase);
\path (agent.south |- resultbend) coordinate (retmid);

\draw[respbranch] (s5.west)
  -- node[lab, above, pos=0.9, yshift=1.5pt, inner sep=1.5pt] {result} (resultbend)
  -- (retmid);
\draw[respbranch] (block.south) -- (errorbend);
\draw[respbranch] (errorbend)
  -- node[lab, above, pos=0.95, yshift=1.5pt, inner sep=1.5pt] {error} (retbase);
\draw[respbranch] (retbase) -- (retmid);
\draw[respflow] (retmid) -- (agent.south);

\end{tikzpicture}%
}
\caption{ChainCaps proxy architecture. The proxy intercepts every
\texttt{tools/call} between the LLM agent and tool server.
Steps~1--2 resolve argument dependencies and compute
$B_\mathrm{agg}=B_\mathrm{ctx}\cap\bigcap_{x\in D}B(x)$.
Step~3 checks whether $\mathrm{Req}(t,a)\in B_\mathrm{agg}$; if not,
it verifies a lineage-bound declassification token before either
forwarding (step~4) or blocking (step~4', dashed).
On the response path, step~5 propagates
$B(y)=\mathrm{Pass}(t)\cap B_\mathrm{agg}$
and updates $B_\mathrm{ctx}$ before returning the result (dashed).
Implementation details are in \S\ref{sec:implementation}.}
\label{fig:architecture}
\end{figure}

%% file: sections/4_evaluation.tex
\section{Evaluation}
\label{sec:evaluation}

\subsection{Setup}
\label{sec:setup}

We evaluate \chaincaps in both live and replay settings. The main experiment uses 82 stress-test tasks: 56 adversarial tasks spanning 12 attack categories and 26 benign tasks spanning 5 workflow types. The adversarial tasks exercise cross-source mixing, scope violations, direct and composed exfiltration, file-based theft, code injection, execution laundering, indirect injection, and shell-mediated leakage. The benign tasks cover ordinary multi-tool workflows where restricted and public values may coexist without policy violation: document and report synthesis, allowed file or database retrieval, scoped web/API lookup, local write/update operations, and approved outbound communication. We run five frontier models from three providers: Claude Sonnet 4, Claude Sonnet 4.6, Claude Opus 4.6, GPT-5.1, and Qwen 3.5. Each model is tested with and without \chaincaps for three fresh-session runs at temperature 0, yielding 82 $\times$ 5 $\times$ 3 = 1{,}230 defended records and 2{,}460 total records across both conditions. We use repeated fresh-session runs because LLM outcomes on deterministic programming-style tasks can vary across invocations even when the task specification is fixed \cite{zhou2026accuracystabilityrepeatedrunreliability}. The live setup uses real tool execution through the MCP proxy. We also replay recorded traces to compare \chaincaps against modeled versions of Fides and PFI under identical source annotations and tool metadata.

We report attack success rate (ASR) on adversarial tasks and benign completion on benign tasks. An adversarial run succeeds if restricted data reaches, or is used to authorize, a sink outside the source's budget. A benign run succeeds if the requested workflow completes without a policy violation or unnecessary block. In replay, we report the fraction of attack traces blocked at the attack step. The live and replay numbers answer different questions: live evaluation measures end-to-end protection in a deployed loop that includes model refusals and alternative plans, while replay isolates the enforcement mechanism on traces that actually attempt the attack.

\subsection{Main Results}
\label{sec:main-results}

Table~\ref{tab:main-results} shows the main deployment result. Without defense, baseline ASR ranges from 25.2\% to 67.8\% depending on the model. With \chaincaps, all five models fall to 0--4.8\% ASR, while benign completion remains between 96\% and 100\%. The result holds across providers and model families, supporting the practical claim that a protocol-level proxy can protect heterogeneous frontier-model agents without model-specific tuning.

\begin{table}[t]
\centering
\caption{Live agent evaluation on 82 tasks with three runs per condition. \chaincaps sharply reduces attack success while preserving benign completion.}
\label{tab:main-results}
\resizebox{\columnwidth}{!}{%
\begin{tabular}{@{}lccccc@{}}
\toprule
\textbf{Metric} & \textbf{Sonnet 4} & \textbf{Sonnet 4.6} & \textbf{Opus 4.6} & \textbf{GPT-5.1} & \textbf{Qwen 3.5} \\
\midrule
ASR without \chaincaps & 67.8\% & 25.2\% & 31.5\% & 53.6\% & 53.0\% \\
ASR with \chaincaps & 4.8\% & 3.6\% & 3.6\% & 2.4\% & 0.0\% \\
Relative reduction & 93\% & 86\% & 89\% & 96\% & 100\% \\
Benign completion w/o \chaincaps & 100\% & 100\% & 100\% & 100\% & 100\% \\
Benign completion w/ \chaincaps & 100\% & 100\% & 96\% & 100\% & 96\% \\
\bottomrule
\end{tabular}%
}
\end{table}

The remaining failures are concentrated rather than broad. Seven of the 12 attack categories are fully neutralized under gold manifests, including cross-source mixing, HTTP scope violations, direct and composed data theft, file-based theft, and code injection. Residual failures cluster in execution laundering (5\% average ASR on the non-Qwen models), indirect injection (8\%), and shell exfiltration (8\%). These are exactly the cases where authority can be re-expressed at a boundary the proxy does not fully see, such as writing a script for later execution or laundering data through shell pipes.

\input{figures/asr-comparison}

\subsection{Comparison with Prior Defenses}
\label{sec:baselines}

Table~\ref{tab:baselines} compares \chaincaps with replayed versions of Fides and PFI. These replay baselines are not full reimplementations of Fides or PFI. Instead, they instantiate the enforcement abstraction most relevant to our setting---scalar value labels for Fides-style IFC and per-function isolation for PFI-style enforcement---under the same recorded traces, source annotations, and tool metadata. The comparison is therefore intended to isolate policy expressiveness under shared traces rather than to evaluate complete systems. Under that interpretation, \chaincaps is consistently stronger across all five models. The margin is large enough to matter in deployment: on GPT-5.1 traces, \chaincaps blocks 59.5\% of attack traces versus 16.1\% for Fides and 10.1\% for PFI. This matches the design argument in \cref{sec:ifc}: scalar IFC labels and per-function isolation each cover part of the problem, but neither directly enforces sink-specific monotonic authority across a tool chain.

\begin{table}[t]
\centering
\caption{Replay-based attack blocking on recorded traces. The Fides and PFI rows model the policy abstraction most relevant to our setting rather than full system reimplementations.}
\label{tab:baselines}
\small
\begin{tabular}{@{}lccccc@{}}
\toprule
\textbf{Defense} & \textbf{S4} & \textbf{S4.6} & \textbf{O4.6} & \textbf{G5.1} & \textbf{Q3.5} \\
\midrule
\chaincaps & \textbf{58.9\%} & \textbf{24.4\%} & \textbf{31.0\%} & \textbf{59.5\%} & \textbf{63.1\%} \\
Fides & 16.7\% & 7.7\% & 10.7\% & 16.1\% & 24.4\% \\
PFI & 8.3\% & 3.0\% & 4.8\% & 10.1\% & 13.7\% \\
\bottomrule
\end{tabular}
\end{table}

We also tested generalization beyond our in-house tasks. On 46 attacks adapted from InjecAgent and ToolEmu style benchmarks \cite{injecagent2024,toolemu2024}, replay through \chaincaps blocks all 46 attacks. End-to-end live protection varies by model because many of these prompts are explicit enough that models self-refuse at different rates, but the enforcement layer blocks every recorded attack step when it is reached.

\subsection{Ablation, Manifest Quality, and Cost}
\label{sec:ablation}

Table~\ref{tab:ablation-manifest} summarizes the two main deployment lessons beyond headline efficacy. First, every major component contributes measurable protection. Removing the meet rule causes the largest drop, reducing attack blocking by 17 percentage points. Removing scope granularity, context tracking, or pass-through attenuation each costs another 7--10 points. This supports the full budget algebra rather than any one heuristic part of it. Second, manifest quality dominates deployment effectiveness. In the manifest-quality sweep, expert manifests achieve 100\% attack blocking with 100\% benign completion, checklist-guided manifests still achieve 90.9\% blocking, and naive manifests fall to 27.3\% blocking and 50\% benign completion. The main practical bottleneck is therefore not overhead or model compatibility; it is policy authoring quality.

\begin{table}[t]
\centering
\caption{Ablation and manifest-quality results. Ablation deltas ($\Delta$) are percentage-point changes in attack blocking relative to the full system on 163 held-out replay traces; all variants reduce blocking. Manifest quality results show end-to-end attack blocking and benign completion in the live manifest sweep.}
\label{tab:ablation-manifest}
\scriptsize
\begin{tabular}{@{}lcc@{}}
\toprule
\multicolumn{3}{@{}l}{\textbf{Component ablation (163-trace replay suite)}} \\
\midrule
\textbf{Variant} & \textbf{Attack block} & \textbf{$\Delta$} \\
\midrule
Full \chaincaps & reference & --- \\
$-$ Meet rule & lower & $-17$ pts \\
$-$ Scope granularity & lower & $-10$ pts \\
$-$ Context budget & lower & $-7$ pts \\
$-$ Pass-through attenuation & lower & $-7$ pts \\
\midrule
\multicolumn{3}{@{}l}{\textbf{Manifest quality (live gold-manifest evaluation)}} \\
\midrule
\textbf{Manifest set} & \textbf{Attack block} & \textbf{Benign completion} \\
\midrule
Gold (expert-authored) & 100.0\% & 100.0\% \\
Careful (checklist-guided) & 90.9\% & 100.0\% \\
Naive (auto-generated) & 27.3\% & 50.0\% \\
\bottomrule
\end{tabular}
\end{table}

Two further results support deployability. The proxy adds 0.13\,ms median latency per tool call (P95: 0.34\,ms), negligible relative to model inference and network time. In an adaptive red-team exercise, \chaincaps blocks 14 of 14 attacks, including attempts at budget washing, staged sink use, and declassification abuse. Together, these results indicate that the mechanism is hard to evade at the application layer once manifests are correct.

\subsection{Threats to Validity}
\label{sec:threats}

The evaluation should be read as stress-test evidence, not as a population estimate of all deployed agents. The 82 tasks intentionally concentrate on composition failures, so the baseline ASR is not a claim about ordinary enterprise traffic. Our replay baselines instantiate the policy abstractions most relevant to this setting rather than full competing systems; they are useful for comparing expressiveness under shared traces, but not for ranking complete implementations. Finally, the strongest results depend on gold manifests and proxy-visible effects. The manifest-quality sweep and residual shell failures show exactly where those assumptions become load-bearing in practice.

\subsection{Takeaways}
\label{sec:takeaways}

Three conclusions follow from the evaluation. First, composition safety is a real deployment problem: baseline ASR of 25--68\% across all five frontier models confirms that per-tool access control leaves a substantial gap on our stress-test suite, independent of which provider or model family is used. Second, monotonic capability attenuation closes that gap more reliably than scalar IFC or per-function isolation in both live and replay settings, with a margin large enough to matter in production (e.g., 58.9\% vs 16.7\% trace blocking on Sonnet 4). Third, the dominant deployment variable is manifest quality, not runtime overhead or model compatibility: the algebra is correct, but enforcement is only as strong as the policy documents describing each tool's source and sink roles. The outstanding engineering challenge for adopting \chaincaps at MCP ecosystem scale is therefore tooling for authoring, linting, and testing manifests, not the runtime mechanism itself.

%% file: figures/asr-comparison.tex
\begin{figure}[t]
\centering
\small
\resizebox{\columnwidth}{!}{%
\begin{tikzpicture}[
  x=1.7cm, y=0.065cm,
  ndbar/.style={fill=red!30, draw=red!60!black, line width=0.4pt},
  ccbar/.style={fill=blue!45, draw=blue!70!black, line width=0.4pt},
]
\draw[gray!40, line width=0.3pt] (0.3,0) -- (6.2,0);
\foreach \y in {0,10,20,30,40,50,60,70} {
  \draw[gray!20, line width=0.3pt] (0.3,\y) -- (6.2,\y);
  \node[anchor=east, font=\scriptsize, text=black!70] at (0.25,\y) {\y};
}
\node[anchor=east, font=\scriptsize, rotate=90] at (-0.15,55) {Attack Success Rate (\%)};

\def\bw{0.22}

\fill[ndbar] (1-\bw,0) rectangle (1,67.8);
\fill[ccbar] (1,0) rectangle (1+\bw,4.8);
\node[above, font=\tiny, text=red!60!black] at (1-\bw/2,67.8) {67.8};
\node[above, font=\tiny, text=blue!70!black] at (1+\bw/2,4.8) {4.8};
\node[below, font=\scriptsize, text=black] at (1,0) {Sonnet 4};

\fill[ndbar] (2-\bw,0) rectangle (2,25.2);
\fill[ccbar] (2,0) rectangle (2+\bw,3.6);
\node[above, font=\tiny, text=red!60!black] at (2-\bw/2,25.2) {25.2};
\node[above, font=\tiny, text=blue!70!black] at (2+\bw/2,3.6) {3.6};
\node[below, font=\scriptsize, text=black] at (2,0) {Sonnet 4.6};

\fill[ndbar] (3-\bw,0) rectangle (3,31.5);
\fill[ccbar] (3,0) rectangle (3+\bw,3.6);
\node[above, font=\tiny, text=red!60!black] at (3-\bw/2,31.5) {31.5};
\node[above, font=\tiny, text=blue!70!black] at (3+\bw/2,3.6) {3.6};
\node[below, font=\scriptsize, text=black] at (3,0) {Opus 4.6};

\fill[ndbar] (4-\bw,0) rectangle (4,53.6);
\fill[ccbar] (4,0) rectangle (4+\bw,2.4);
\node[above, font=\tiny, text=red!60!black] at (4-\bw/2,53.6) {53.6};
\node[above, font=\tiny, text=blue!70!black] at (4+\bw/2,2.4) {2.4};
\node[below, font=\scriptsize, text=black] at (4,0) {GPT-5.1};

\fill[ndbar] (5-\bw,0) rectangle (5,53.0);
\fill[ccbar] (5,0) rectangle (5+\bw,0.5); 
\node[above, font=\tiny, text=red!60!black] at (5-\bw/2,53.0) {53.0};
\node[above, font=\tiny, text=blue!70!black] at (5+\bw/2,1.5) {\textbf{0}};
\node[below, font=\scriptsize, text=black] at (5,0) {Qwen 3.5};

\fill[ndbar] (4.0,66) rectangle (4.35,69);
\node[anchor=west, font=\scriptsize] at (4.4,67.5) {No Defense};
\fill[ccbar] (4.0,59) rectangle (4.35,62);
\node[anchor=west, font=\scriptsize] at (4.4,60.5) {\textsc{ChainCaps}};

\end{tikzpicture}%
}
\caption{Attack success rate across five frontier models. Without
defense, ASR ranges from 25\% to 68\%. With \chaincaps, all tested
models fall to $\leq$5\% ASR (Qwen~3.5 reaches 0\%), corresponding to
an 86--100\% relative reduction on this stress-test suite.}
\label{fig:asr-comparison}
\end{figure}
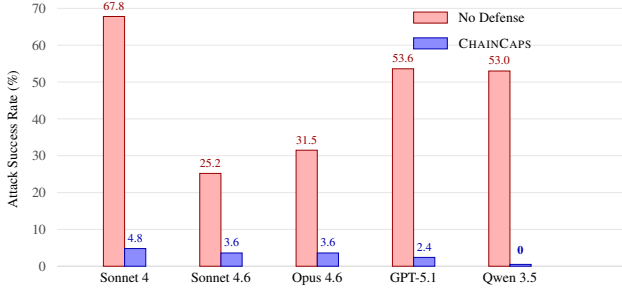

%% file: sections/5_conclusion.tex
\section{Conclusion}
\label{sec:conclusion}

Deployed tool-using agents need more than per-tool authorization. When models assemble workflows at runtime, locally safe calls can compose into unsafe end-to-end behavior. \chaincaps addresses that gap with a narrow but enforceable invariant: authority attached to data attenuates monotonically through the composition chain. Across five frontier models and an 82-task stress-test suite, this invariant reduces attack success from 25--68\% to 0--4.8\% while preserving 96--100\% benign completion and adding sub-millisecond overhead.

Explicit-flow composition safety is enforceable today in MCP-style ecosystems when manifests are accurate and the relevant calls pass through the proxy. The binding constraint on deployment effectiveness is manifest quality, not the runtime algebra. Three concrete next steps follow from our results: (1) automated manifest generation that infers source and sink roles from tool documentation and observed behavior; (2) a manifest linter and test harness that operators can run in CI before deploying new tools; (3) extension of budget semantics to agent-tool protocols beyond MCP, including direct API calls and browser-based agents. Residual failures in script indirection and shell-pipe laundering require a different class of solution --- either treating shell execution as a composite sink with OS-level inspection, or coordinating with network-layer monitoring to observe effects the proxy cannot directly see. Monotonic capability attenuation closes the composition safety gap at the application layer; closing the remaining boundary failures will require enforcement that reaches below it.

More broadly, our results suggest that agent safety in the wild cannot be reduced to model refusal or per-tool authorization alone. As agents increasingly assemble workflows dynamically across heterogeneous tools, deployment systems need runtime invariants over the composed workflow itself. \chaincaps provides one such invariant for explicit dataflow: values may lose authority as they move through the chain, but composition should not allow them to regain authority that their sources never had.

\paragraph{Code and artifact availability.}
The public artifact is available at \url{https://github.com/Jxcup/chaincaps-code} under an MIT license.